\documentclass[sigconf]{acmart}

\usepackage{subfigure}  % 子图
\usepackage{multirow}   % 表格
\usepackage{amsthm,amsmath,amsfonts}  % 数学
\newtheorem{theorem}{Theorem}
\usepackage{pifont}
\newcommand{\cmark}{\ding{51}}%
\newcommand{\xmark}{\ding{55}}%
\newcommand\ours{PGSP}

\AtBeginDocument{%
  \providecommand\BibTeX{{%
    \normalfont B\kern-0.5em{\scshape i\kern-0.25em b}\kern-0.8em\TeX}}}

\copyrightyear{2023}
\acmYear{2023}
\setcopyright{acmlicensed}\acmConference[WWW '23]{Proceedings of the ACM Web Conference 2023}{May 1--5, 2023}{Austin, TX, USA}
\acmBooktitle{Proceedings of the ACM Web Conference 2023 (WWW '23), May 1--5, 2023, Austin, TX, USA}
\acmPrice{15.00}
\acmDOI{10.1145/3543507.3583466}
\acmISBN{978-1-4503-9416-1/23/04}

\begin{document}

\title{Personalized Graph Signal Processing for Collaborative Filtering}

\author{Jiahao Liu}\authornote{Also with Shanghai Key Laboratory of Data Science, Fudan University, Shanghai,
China.}
\affiliation{%
  \institution{School of Computer Science \\Fudan University}
  \city{Shanghai}
  \country{China}}
\email{jiahaoliu21@m.fudan.edu.cn}

\author{Dongsheng Li}
\affiliation{%
  \institution{Microsoft Research Asia}
  \city{Shanghai}
  \country{China}}
\email{dongshengli@fudan.edu.cn}

\author{Hansu Gu}\authornote{Corresponding author.}
\affiliation{%
  \city{Seattle}
  \country{United States}}
\email{hansug@acm.org}

\author{Tun Lu}\authornotemark[1]\authornotemark[2]
\affiliation{%
  \institution{School of Computer Science \\Fudan University}
  \city{Shanghai}
  \country{China}}
\email{lutun@fudan.edu.cn}

\author{Peng Zhang}\authornotemark[1]
\affiliation{%
  \institution{School of Computer Science \\Fudan University}
  \city{Shanghai}
  \country{China}}
\email{zhangpeng_@fudan.edu.cn}

\author{Li Shang}\authornotemark[1]
\affiliation{%
  \institution{School of Computer Science \\Fudan University}
  \city{Shanghai}
  \country{China}}
\email{lishang@fudan.edu.cn}

\author{Ning Gu}\authornotemark[1]
\affiliation{%
  \institution{School of Computer Science \\Fudan University}
  \city{Shanghai}
  \country{China}}
\email{ninggu@fudan.edu.cn}

\begin{abstract}
% 对应 Introduction 第1段
% 引出CF、GSP-based CF
The collaborative filtering (CF) problem with only user-item interaction information can be solved by graph signal processing (GSP), which uses low-pass filters to smooth the observed interaction signals on the similarity graph to obtain the prediction signals.
% 对应 Introduction 第2段
% 介绍现有的不足
However, the interaction signal may not be sufficient to accurately characterize user interests and the low-pass filters may ignore the useful information contained in the high-frequency component of the observed signals, resulting in suboptimal accuracy.
% 过渡句
To this end, we propose a personalized graph signal processing (PGSP) method for collaborative filtering.
% 对应 Introduction 第3段
% 个性化交互信号 和 增广相似度图
Firstly, we design the personalized graph signal containing richer user information and construct an augmented similarity graph containing more graph topology information, to more effectively characterize user interests.
% 对应 Introduction 第4段
% 混合频率图滤波器
Secondly, we devise a mixed-frequency graph filter to introduce useful information in the high-frequency components of the observed signals by combining an ideal low-pass filter that smooths signals globally and a linear low-pass filter that smooths signals locally.
% 对应 Introduction 第5段
% PGSP Pipeline
Finally, we combine the personalized graph signal, the augmented similarity graph and the mixed-frequency graph filter by proposing a pipeline consisting of three key steps: pre-processing, graph convolution and post-processing.
% 这句总要放最后
Extensive experiments show that PGSP can achieve superior accuracy compared with state-of-the-art CF methods and, as a nonparametric method, PGSP has very high training efficiency.
\end{abstract}

% http://dl.acm.org/ccs.cfm
\begin{CCSXML}
<ccs2012>
<concept>
<concept_id>10002951.10003260.10003261.10003271</concept_id>
<concept_desc>Information systems~Personalization</concept_desc>
<concept_significance>500</concept_significance>
</concept>
</ccs2012>
\end{CCSXML}

\ccsdesc[500]{Information systems~Personalization}

\keywords{recommendation, collaborative filtering, graph signal processing}

\maketitle

\section{Introduction}
% 第1段：介绍CF，用GSP解决CF问题
Collaborative filtering (CF) is one of the most popular techniques in recommender systems, which predicts the interaction possibility between users and items through interaction history~\cite{Herlocker,Sarwar,sarwar2000application,Koren09,he2017neural,LightGCN}.
Recently, graph signal processing (GSP) has been applied in CF tasks~\cite{Narang13a, Narang13b, LSB, VM, Diffusion, GFCF}, which can transform the interaction signals from spatial domain to spectral domain through Graph Fourier Transform (GFT) to capture the structure information of the whole similarity graph.
Generally, GSP-based CF methods first construct a similarity graph between items, then use a low-pass filter defined on the similarity graph to smooth the users' interaction signals to obtain the prediction signals.
Several recent studies~\cite{Huang18, Revisit, GFCF} show that these methods are related to both neighborhood-based CF methods~\cite{Herlocker,Sarwar} and GCN-based CF methods~\cite{LightGCN}.

% 第2段：讲之前的GSP-based方法
% 这里只写输入信号和高频信息存在的问题
% 分别在第4、5段给出我们的解决方法
% 不足：
Existing GSP-based CF methods rely on the following two assumptions:
% （1）仅使用了观测到的用户交互信号
1) a user's personalized preference is completely characterized by his/her interaction history and  
% （2）仅使用了低频信息
2) the low-frequency information in the interaction signal is enough to predict user preferences.
% 导致：
However, due to data sparsity issues in many recommender systems, sparse historical interactions may not be enough to accurately describe user preferences.
% 【新加的】
Besides, only using low-frequency information containing common user interests may not be personalized enough due to neglecting the useful high-frequency information containing unique user interests.

% 第3段：介绍我们的工作-增广相似度图，个性化用户信号（顺带介绍构图动机）
% 【这段重写】
% personalized graph signal
In this paper, we propose the {\em personalized graph signal}, a signal with richer user information, to describe users more accurately.
Instead of only using the interaction information of users, the personalized graph signals use both the similarity information between users-items and the similarity information between users-users to describe users.
% augmented similarity graph
Moreover, we construct an \textit{augmented similarity graph} to utilize the signals more effectively.
The augmented similarity graph contains not only the similarity information between items-items, but also the similarity information between users-items and the similarity information between users-users.
With the personalized graph signal and the augmented similarity graph, users who are similar to the target user can amplify the signals of the items they have interacted with, which often contain the potential interests of the target user.
We construct the similarity relationship from the perspective of random walk, which shares some of the motivations in Node2Vec~\cite{Node2Vec}.

% 第4段：介绍我们的工作-高频信息（顺带介绍高通滤波器）
% 对应实验5.4，解决第3段提出的问题
% 观测信号是混合信号
In the real world, the observed interaction signal of each user is a mixture of the real preference signal dominated by low-frequency information and the noise signal dominated by high-frequency information.
% 真实偏好由全局平滑的普遍偏好信号和局部平滑的个性化平滑组成
Specifically, the real preferences of each user are composed of globally smooth signals that reflect the general user preferences and locally smooth but globally rough signal that reflects the personalized user preferences.
% 第一种图滤波器
In order to obtain these two kinds of signals, we propose a \textit{mixed-frequency graph filter}, which is a combination of an ideal low-pass filter that smooths signals globally and a linear low-pass filter that smooths signals locally.
By controlling the proportion of the two signals, we can obtain more accurate recommendation results than using only an ideal low-pass graph filter.

% 第5段：介绍我们的工作
% 其他工作，不与第3段对应，包括：
% 框架
% 实验5.5：规范化矩阵，实验5.6截止频率（不在这里强调了）
% 框架
Finally, we propose the \textit{personalized graph signal processing} (PGSP) method by designing a pipeline to combine the personalized graph signal, the augmented similarity graph and the mixed-frequency graph filter to achieve higher accuracy. % without compromising diversity.
% 分步骤介绍这个框架
PGSP consists of three key steps:
{\em 1) pre-processing}, in which we construct the personalized graph signal by concatenating user similarity information with user-item interaction signal;
{\em 2) graph convolution}, in which we use the proposed mixed frequency graph filter to obtain both globally smooth signals and locally smooth but globally rough signals and
{\em 3) post-processing}, in which we recover the predicted interaction signal.
Extensive experiments show that PGSP has superior performance in prediction accuracy.
Meanwhile, we analyze the role of the mixed-frequency graph filter and show that, as a nonparametric method, PGSP has very high training efficiency.

% 第6段：贡献
The main contributions of this work are summarized as follows:
\begin{itemize}
% Intro第4段：增广相似度图，个性化用户信号
\item We propose the personalized graph signal and the augmented similarity graph for GSP-based CF, which can characterize user interests more effectively and realize more personalized recommendations.
% Intro第5段：混合频率滤波器
\item We reveal the effectiveness of high-frequency components in the observed signal, and find that controlling the ratio of globally smooth signal to locally smooth signal by the proposed mixed-frequency graph filter can further improve the accuracy.
% Intro第6段：框架
\item We propose the personalized graph signal processing (PGSP) method by combining the personalized graph signal, the augmented similarity graph and mixed-frequency graph filters. Experiments %on three widely used datasets 
show that PGSP can achieve higher accuracy compared with state-of-the-art CF methods.
\end{itemize}

\section{Related Work}
% 介绍GSP，引出GSP-based CF
GSP aims to develop tools for processing data defined on irregular graph domains, including sampling, filtering and graph learning~\cite{Chung97,GSP}.
%There are rich theories and a large number of mathematical tools in the field of GSP~\cite{Chung97}. 
Next, we briefly summarize GSP-based CF works.

% 1、图信号重构类方法
Most GSP-based methods treat the CF problem as graph signal reconstruction~\cite{chen2021scalable}.
~\citet{Narang13a} obtain the interpolated signal by projecting the input signal into the appropriate bandlimited graph signal space. Then, ~\citet{Narang13b} extend the previous work to an iterative method, which has higher computational efficiency and takes into account the reconstruction error relative to the known signals.
~\citet{LSB} reconstruct the graph signal by reweighting the sampled residuals for different vertices or propagating the sampled residuals in their respective local sets.
~\citet{VM} formulate graph signal recovery as an optimization problem and provide a general solution through the alternating direction methods of multipliers.
~\citet{Diffusion} interpret the smoothing of the signal by low-pass filtering operation as the diffusion of temperature and prove that the diffused signals are stable to perturbations in the underlying network.

% 2、与神经网络相关的方法
Some CF methods introduced GSP into neural networks.
SpectralCF~\cite{SpectralCF} constructs a deep model for learning in the spectral domain, which uses rich connection information to alleviate the cold start problem.
~\citet{FP} extend SpectralCF by directly mining the association rules between the target user and the target item through FP-Growth~\cite{Han2000}, and then solve the problems of cold start and data sparsity.
AGE~\cite{AGE} proposes an attributed graph embedding framework, which consists of a carefully-designed Laplacian smoothing filter and an adaptive encoder.
~\citet{Yu20} propose a low-pass graph filter to remove the noise and reduce the complexity of graph convolution in an unscathed way.
AGCN~\cite{AGCN} further reduces complexity with Chebyshev polynomial graph filters.

% 3、与其他方法的关联
Recently, some works have shown that there are close relationships between GSP-based CF methods and other CF methods~\cite{liu2022parameter,xia2022fire}.
~\citet{Huang18} demonstrate that neighborhood-based methods can be modeled as a specific band-stop graph filter and low-rank matrix completion can be viewed as band-limited interpolation algorithms.
~\citet{Revisit} find that graph neural networks only perform low-pass filtering on feature vectors and do not have the non-linear manifold learning property, which means that the graph structure only provides a means to denoise the data. 
~\citet{Simplify} reduce the complexity of GCN through successively removing nonlinearities and collapsing weight matrices between consecutive layers, and show that the resulting linear model corresponds to a fixed low-pass filter.
GF-CF~\cite{GFCF} develops a unified graph convolution-based framework for CF and proves that many existing CF methods are special cases of the framework which correspond to different low-pass filters in GSP.

\section{Preliminaries}
This section introduces the background knowledge of GSP and the motivation for using GSP in CF tasks.

\subsection{Graph Signal Processing}
% 基本算子介绍
Given a graph $\mathcal{G}=\{\mathcal{V},\mathcal{E}\}$, where $\mathcal{V}=\{v_1, v_2,...,v_n\}$ is the vertex set with $n$ nodes, $\mathcal{E}$ is the edge set. The topology structure of graph $\mathcal{G}$ can be denoted by a weighted adjacency matrix $W=\{w_{ij}\}\in\mathbb{R}^{n\times n}$, where $w_{ij}>0$ if $(v_i,v_j)\in \mathcal{E}$, indicating the similarity between $v_i$ and $v_j$, otherwise $w_{ij}=0$. $D=diag(d_1,d_2,...,d_n)\in\mathbb{R}^{n\times n}$ denotes the degree matrix of $W$, where $d_i=\sum_{v_j\in\mathcal{V}}w_{ij}$ is the degree of node $v_i$. The graph Laplacian matrix is defined as $L=D-W$.

% 空域信号
The signal on the node $v_i$ is defined as a mapping $x_i:\mathcal{V}\rightarrow\mathbb{R}$. We can take $\pmb{x}\in\mathbb{R}^n$ as a graph signal where each node is assigned with a scalar. 
The energy of the graph signal is defined as $E(\pmb{x})=||\pmb{x}||^2$. The smoothness of the graph signal can be measured by the total variation as follows:
\begin{displaymath}
\textstyle
  TV(\pmb{x})=\pmb{x}^TL\pmb{x}=\sum_{(v_i,v_j)\in \mathcal{E}}w_{ij}(x_i-x_j)^2.
\end{displaymath}
The normalized total variation of $\pmb{x}$ can be calculated with the Rayleigh quotient as follows:
\begin{displaymath}
  Ray(\pmb{x})=\frac{TV(\pmb{x})}{E(\pmb{x})}=\frac{\pmb{x}^TL\pmb{x}}{\pmb{x}^T\pmb{x}}=\frac{\sum_{(v_i,v_j)\in \mathcal{E}}w_{ij}(x_i-x_j)^2}{\sum_{v_i\in \mathcal{V}}x_i^2}.
\end{displaymath}
% GFT & IGFT
As $L$ is real and symmetric, its eigendecomposition is given by $L=U\Lambda U^T$ where $\Lambda=diag(\lambda_1,\lambda_2,...,\lambda_n),\lambda_1\le \lambda_2\le...\le \lambda_n$, and  $U=(\pmb{u}_1,\pmb{u}_2,...,\pmb{u}_n)$ with $\pmb{u}_i\in\mathbb{R}^{n}$ being the eigenvector for eigenvalue $\lambda_i$. We call $\tilde{\pmb{x}}=U^T\pmb{x}$ as the graph Fourier transform of the graph signal $\pmb{x}$ and its inverse transform is given by $\pmb{x}=U\tilde{\pmb{x}}$.

% 频域信号
GFT transfers the graph signal from the spatial domain to the spectral domain.
Rayleigh quotient can be transformed into the spectral domain as follows:
\begin{equation}
  \label{eq.Ray}
  Ray(\pmb{x})=\frac{\pmb{x}^TL\pmb{x}}{\pmb{x}^T\pmb{x}}=\frac{\pmb{x}^TU\Lambda U^T\pmb{x}}{\pmb{x}^TUU^T\pmb{x}}=\frac{\tilde{\pmb{x}}^T\Lambda \tilde{\pmb{x}}}{\tilde{\pmb{x}}^T\tilde{\pmb{x}}}=\frac{\sum_{v_i\in \mathcal{V}}\lambda_i\tilde{x}^2_i}{\sum_{v_i\in \mathcal{V}}\tilde{x}^2_i}.
\end{equation}
The graph filter $\mathcal{H}$ is defined as:
\begin{equation}
  \label{eq.filter}
  \mathcal{H}=Udiag(h(\lambda_1),h(\lambda_2)...,h(\lambda_n))U^T,
\end{equation}
where $h(\cdot)$ is the frequency response function. The graph convolution of an input signal $\pmb{x}$ and the filter $\mathcal{H}$ is defined as follows:
\begin{displaymath}
  \pmb{y}=\mathcal{H}\pmb{x}=Udiag(h(\lambda_1),h(\lambda_2)...,h(\lambda_n))U^T\pmb{x}.
\end{displaymath}
In summary, the original graph signal $\pmb{x}$ is first transformed into Fourier space by Fourier basis $U^T$, then the spectral domain signal is processed by filter $h(\cdot)$, and finally inversely transformed into spatial domain by $U$.

\subsection{Motivation of Using GSP in CF}
\label{sec:motivation}
% 真实偏好的总变差小
CF tasks aim to predict the real preference matrix according to the observed interaction matrix.
If we have constructed the correct similarity matrix between items and the input signal is a user's real preference signal which indicates the user's preference for each item, the total variation of the graph signal can be discussed in two cases:
1) if two items $i$ and $j$ $(i\not=j)$ are similar, i.e., $w_{ij}$ is large, the user's preferences for these two items will be similar, which means that $(x_i-x_j)^2$ should be small. This situation will not lead to excessive $TV(\pmb{x})$;
2) if two items $i$ and $j$ $(i\not=j)$ are dissimilar, i.e., $w_{ij}$ is small, the user often has different preferences for the two items, which means that $(x_i-x_j)^2$ should be large.
This also does not cause $TV(\pmb{x})$ to be too large because their similarity $w_{ij}$ is small.
In conclusion, if the real preference signal is used as input, the total variation should be small.

% 观测到的是掺杂噪声的混合信号
However, due to the exposure noise and quantization noise in the observed interaction matrix \cite{Yu20}, the total variation becomes larger when the input signal is the observed user interaction signal.
Fig. \ref{Fig.noise.1} shows the real preference matrix, each row of which is a user's preference signal.
Exposure noise means that users can not encounter all items, but only a subset of them.
These items that are not exposed to users may or may not be liked by the users.
Fig. \ref{Fig.noise.2} shows the observed change in the user preference matrix due to exposure noise.
Users will choose whether to interact with the items exposed to them based on their preferences.
Finally, our observed user-item ratings are binary variables, which are no longer the real preference values, due to quantization noise.
Fig. \ref{Fig.noise.3} shows the change of user preference matrix observed by us due to quantization noise, that is, the final observed interaction matrix.

\begin{figure}
\centering
% User preference matrix
\subfigure[Real preference matrix]{
    \label{Fig.noise.1}
    \includegraphics[width=0.3\linewidth]{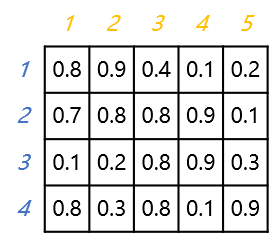}
}
% exposure noise
\subfigure[With exposure noise]{
    \label{Fig.noise.2}
    \includegraphics[width=0.3\linewidth]{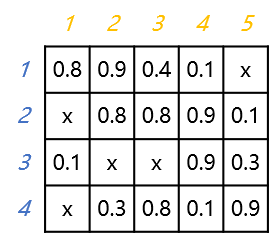}
}
% quantization noise
\subfigure[With quantization noise]{
    \label{Fig.noise.3}
    \includegraphics[width=0.3\linewidth]{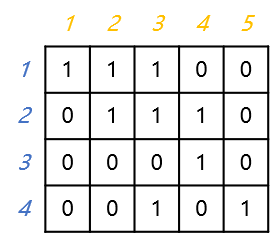}
}
\caption{Relationship between the real preference matrix and observed preference matrix by considering noises.}
\label{fig:three}
\end{figure}

% 低通滤波器解决方案
% 逻辑：
% 1、混合信号由低频、高频构成
% 2、通过图卷积可以分离两类信号
From the perspective of the spectral domain, the real preference signal is a low-frequency signal, and the two kinds of noise are high-frequency signals.
The observed interaction signal is the fusion of low-frequency signal and high-frequency signal.
We can effectively solve the CF problem by graph convolution, the key of which is to design a graph filter determined by the frequency response function.
For Eq. (\ref{eq.Ray}), take $\pmb{x}=\pmb{u}_i$, we can get $Ray(\pmb{u}_i)=\lambda_i$, indicating that the eigenvector corresponding to the small eigenvalue is smoother, and this smoothing is smooth in the sense of the whole graph, e.g., globally smooth.
When restoring the spectral domain signal to the spatial domain, if only the smoother eigenvectors are used as the basis vectors, the reconstructed spatial signal will have a lower frequency in the sense of the whole graph.

Fig. \ref{Fig:pre}(i) illustrates the item similarity graph used in previous works.
we use $1$ to represent strong node signal, and $0$ to represent weak node signal.
The input graph signal is $(1,0,0)$, which makes item $1$ have a strong node signal, as shown in Fig. \ref{Fig:pre}(ii).
After graph convolution, the graph signal becomes smooth, that is, similar items have similar graph signals.
This improves the signal of item $2$ and then recommends it to the user $a$, as shown in Fig. \ref{Fig:pre}(iii).

However, we believe that the real user preference and these two kinds of noise can not be simply separated from the spectral domain by frequency.
Specifically, globally smooth signals are dominated by general preferences which may not be enough to capture personalized preferences.
Therefore, it is necessary to introduce locally smooth but globally rough signals to capture personalized preferences, which can help to improve recommendation quality.

\begin{figure}[tb!]
  \centering
	\includegraphics[width=0.98\linewidth]{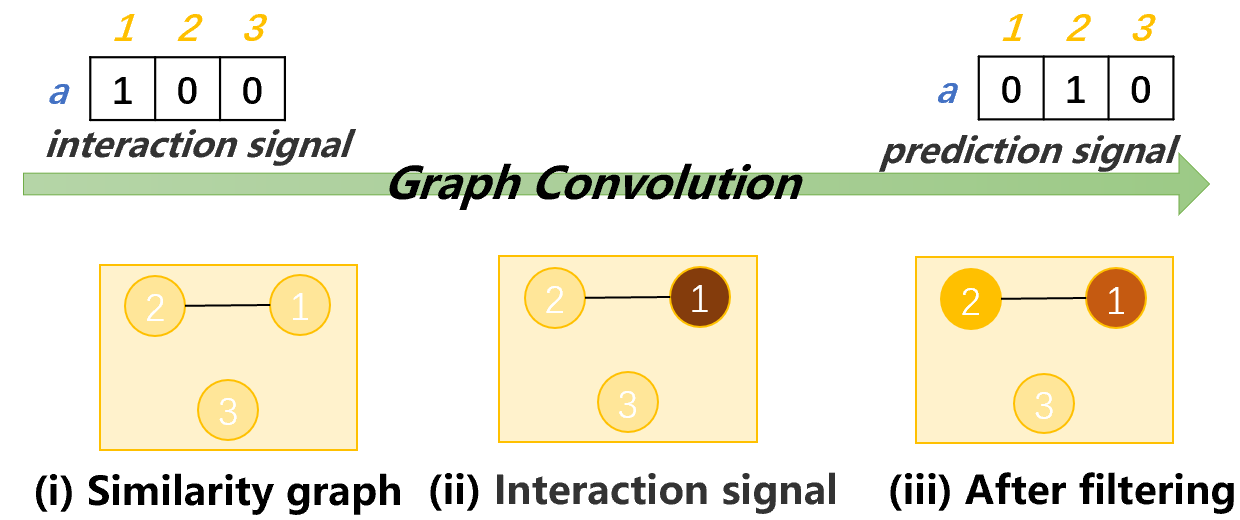}
    \caption{The previous GSP-based CF method. The interaction signal is directly input into the similarity graph.}
    \vspace{-1em}
    \label{Fig:pre}
\end{figure}
\section{The PGSP Method}
In this section, we propose the personalized graph signal processing (PGSP) method which consists of three key components:
\begin{itemize}
  \item {\em Personalized Graph Signal}, which has richer user-personalized information to describe users more accurately;
  \item {\em Augmented Similarity Graph}, which has more graph topology information to utilize the personalized graph signals more effectively;
  \item {\em Mixed-Frequency Graph Filter}, which allows globally smooth signals that reflect the general user preferences and locally smooth but globally rough signals that reflect the personalized user preferences to pass through.
\end{itemize}
We combine the above three components and propose the {\em PGSP Pipeline}, which consists of pre-processing, graph convolution and post-processing to achieve higher accuracy.
In the following discussion, we  suppose that there are $m$ users and $n$ items, and their interaction matrix is $R\in\{0,1\}^{m\times n}$, where $R_{ij}=1$ indicates that user $i$ has interacted with item $j$.

% 4.1 个性化交互信号和增广相似度矩阵
\subsection{Personalized Graph Signal and Augmented Similarity Graph}
In this section, we first introduce how to construct the similarity matrix when only relying on the graph structure from the perspective of random walk, and then propose personalized graph signal and augmented similarity graph.
\subsubsection{Construction of Similarity Graph}
% 一、回答“为什么”
% （1）介绍相似度的三个决定因素
% 当依赖且仅依赖图的拓扑结构，我们可以从Node2Vec的设计中总结出，节点间的相似度由三个因素决定
When only considering the graph topology, we conclude from the design of \textit{Node2Vec}~\cite{Node2Vec} that the similarity of two nodes is determined by three factors:
% 节点间路径数
1) the number of paths between two nodes,
% 路径长度
2) the length of the paths and
% 路径中间节点的度数
3) the degree of the nodes that compose each of the paths.
% 通常，具有高相似性的两个节点间具有大量的长度短的路径，并且形成这些路径的节点的度不会太大。
In {random walk}, there are usually a large number of short paths between two nodes with high similarity, and the degree of the nodes composing each of the paths should not be too large (otherwise the probability of walking through each path will be low).

% （2）与协同过滤任务结合
% 对于只有用户项交互信息的协同过滤问题，我们所拥有的只是交互图的结构，因此在构建相似图时应该考虑这三个因素。
For the CF problem with only user-item interaction information, all we have is the structure of the interaction graph, so these three factors should be taken into account as prior knowledge when constructing the similarity graph.
% 用户-物品交互图是一个二部图，这意味着用户与物品直接连接，而两个用户/物品之间一定会经过一个物品/用户
The user-item interaction graph is a bipartite graph, which means that users are directly connected to items, and paths between two users/items must pass through intermediate items/users.
% 因此，用户-用户、物品-物品，用户-物品相似度都需要同时考虑用户节点和物品节点的度
Therefore, for user-user similarity, item-item similarity, and user-item similarity, we need to consider the degrees of user nodes and item nodes at the same time.

% 二、回答“是什么”和“怎么做”
% （1）“用户-物品”相似度矩阵的构建（规范化观测矩阵）
% 交互矩阵指示了用户节点与物品节点的长度为1的路径，路径数为0或1
The interaction matrix indicates the existence of a path with a length of 1 between a user node and an item node.
% 原始交互矩阵R不能用于描述用户与物品的相似度，因为它没有考虑节点的度
The original interaction matrix $R$ can not be used to describe the similarity between users and items, because it does not consider the node degree.
% 我们考虑了节点的度，并获得了归一化交互矩阵，该矩阵可用于表示用户和物品之间的相似性
We take into account the degree of nodes and obtain the normalized interaction matrix, which can be used to represent the similarity between users and items as follows:
\begin{equation}
  S_{UI}=(D_U)^{-1/2}R(D_I)^{-1/2}.
\end{equation}
The $D_U$ and $D_I$ are diagonal matrix, where $(D_U)_{ii}=\sum_{j=1}^nR_{ij},i=1,...,m$ and $(D_I)_{jj}=\sum_{i=1}^mR_{ij},j=1,...,n$.
% 从随机游走的角度，边的权重完全取决于由边连接的节点的度
From the perspective of random walk, the weight of an edge completely depends on the degree of the nodes connected by the edge.
% 解释为什么这样可以
% 一个用户交互过的物品越多，单个物品就越不能描述这个用户，用户与他交互的项目之间的相似性就越低。
The more items that a user has interacted with, the less a single item can be capable of describing the user, and the lower the similarity between the user and the items.
The same is true from the perspective of an item.

% （2）“物品-物品”和“用户-用户”相似度矩阵的构建
% 伴随规范化交互矩阵的建立，我们分别可以计算用户间相似度矩阵和物品间相似度矩阵，如下
Based on the normalized interaction matrix, we can calculate the user-user similarity matrix and item-item similarity matrix, respectively, as follows:
\begin{equation}
  \label{eq.SU}
  S_U = S_{UI}S_{UI}^T,\quad S_I=S_{UI}^TS_{UI}.
\end{equation}
% 解释为什么这样可以
% 从随机游走的角度看，A的构建模拟了从用户节点转移到物品节点（由B表示），然后再次转移到用户节点（由C表示）的两步游走过程。
From the perspective of random walk, the construction of $S_U$ simulates the two-step walk process of transferring from the user nodes to the item nodes (represented by $S_{UI}$) and then to the user nodes (represented by $S_{UI}^T$).
% 两个用户节点之间长度为2的路径越多并且这些路径的权重越高，两个用户节点之间的相似性也越高。
More paths with a length of 2 between two user nodes and higher weights of these paths will indicate higher similarity between the two user nodes.
% 对于A有类似的结论。
There is a similar conclusion for the item-item similarity matrix $S_I$.

% 其他构建方法的不足
% 我们在这里强调共现关系和余弦相似度是不适合用于在这种场景下计算相似度的
We emphasize that co-occurrence relationship and cosine similarity are not suitable for calculating similarity in this scenario.
% 当利用它们，计算方式如下所示
In detail, their calculations are as follows:
\begin{displaymath}
\begin{aligned}
  S^{co}_U=RR^T,\quad
  S^{cos}_U=D_U^{-1/2}RR^TD_U^{-1/2}.\\ S^{co}_I=R^TR,\quad
  S^{cos}_I=D_I^{-1/2}R^TRD_I^{-1/2}.\\
\end{aligned}
\end{displaymath}
% 显然，这两种方式没有充分考虑节点的度对随机游走的影响
Obviously, these two methods do not fully consider the influence of node degree on random walk.

\subsubsection{Construction of Personalized Graph Signal}
% 一、回答“为什么”
% 仅user-item不够用
Due to data sparsity issues in many recommender systems, sparse historical interactions may not be enough to describe user preferences accurately.
% 对于用户，除了user-item这一直接关系，还有user-user这一间接关系
For each user, in addition to the direct interactions with items, there is also a collaborative relationship with other users, which is an indirect relationship.
% 利用这两种关系来描述用户
We combine these two kinds of information to describe users more accurately.

% 二、回答“是什么” 和“怎么做”
% 包括了“阈值”“权重”“拼接”
For each user, we concatenate his/her similarities with other users with his/her interaction signal to obtain the personalized graph signal as follows:
\begin{equation}
    \label{r-tilde}
    \tilde{R}=S_U||R,
\end{equation}
where $||$ is the concatenation operation.

\subsubsection{Construction of Augmented Similarity Graph}
% 一、回答“为什么”
% 陈述事实：交互矩阵是用户行为的直接体现，描述用户-物品的关系；相似度矩阵是对交互信息的压缩，描述物品-物品或用户-用户间的关系
The interaction matrix is the direct embodiment of users' behaviors and describes the relationship between users and items.
The similarity matrix is the compression of interaction information to describe the relationship between item-item and user-user.
% 两个动机：
There are two motivations for the augmented similarity graph.
% （1）这三种关系都很重要。利用用户-物品关系可以实现用户点信号直接影响物品点信号；利用用户-用户关系和物品-物品关系可以分别实现用户点信号之间的相互影响和物品点信号之间的相互影响
Firstly, all these relationships are useful for GSP-based CF methods.
Using the user-item relationship, the signals of user nodes can directly affect the signals of item nodes, and the existence of the item-item/user-user relationship makes it possible for the signals between item/user nodes to interact with each other.
% （2）个性化交互信号是m+n维向量，需要一个能容纳更多点信号的相似度图
Secondly, the personalized graph signal is an $m+n$-dimensional vector, which requires a similarity graph that can accommodate more node signals.

% 二、回答“是什么”和“怎么做”
We construct the augmented similarity graph, which contains all the above three relationships as follows:
\begin{equation}
\label{eq.ASG}
  A=\left[ \begin{matrix} S_U & S_{UI}\\ S_{UI}^T & S_I\end{matrix} \right].
\end{equation}
% 总结：一个更合理的从随机游走的角度出发构建的相似度矩阵构建方式用以刻画user-item、user-user和item-item关系；一个拥有更丰富信息的个性化交互信号用以更准确刻画用户；一个拥有更多拓扑信息的增广相似度矩阵用以更有效地处理个性化交互信号

Overall, our method is beneficial to CF in the following aspects:
1) a more reasonable similarity matrix constructed from the perspective of random walk is used to describe the relationship of users-items, users-users and items-items,
2) a personalized graph signal with richer information is used to describe users more accurately, and
3) an augmented similarity matrix with more topological information is used to process personalized graph signals more effectively.

% 4.2 混合频率滤波器
\subsection{Mixed-Frequency Graph Filter}
In this section, we propose a mixed-frequency graph filter, which can capture globally smooth signals and locally smooth but globally rough signals simultaneously. 

% 理想低通滤波器
With the definition of the augmented similarity graph, we can define the corresponding Laplacian matrix $L=I-A$, where $I$ is an identity matrix.
The frequency response function of the ideal low-pass filter is $h(\lambda)= 1$ if $\lambda \le \lambda_k$ and $h(\lambda)=0$ otherwise, 
where $\lambda_k$ is the cut-off frequency of the low-pass filter. Putting it into  Eq. (\ref{eq.filter}), we can obtain the ideal low-pass filter as follows:
\begin{displaymath}
 \mathcal{H}_{ideal} =Udiag(\overbrace{1,...,1}^k,\overbrace{0,...,0}^{m+n-k})U^T=U_kU_k^T,
\end{displaymath}
where $U$ is a matrix composed of the eigenvectors of $L$, $U_k$ is the first $k$ columns of $U$ corresponding to the smaller eigenvalues.
This indicates that the construction of the ideal low-pass filter does not need complete eigendecomposition of the matrix, but only needs to calculate the eigenvectors corresponding to the smaller eigenvalues.

\begin{figure}[tb!]
  \centering
	\includegraphics[width=0.95\linewidth]{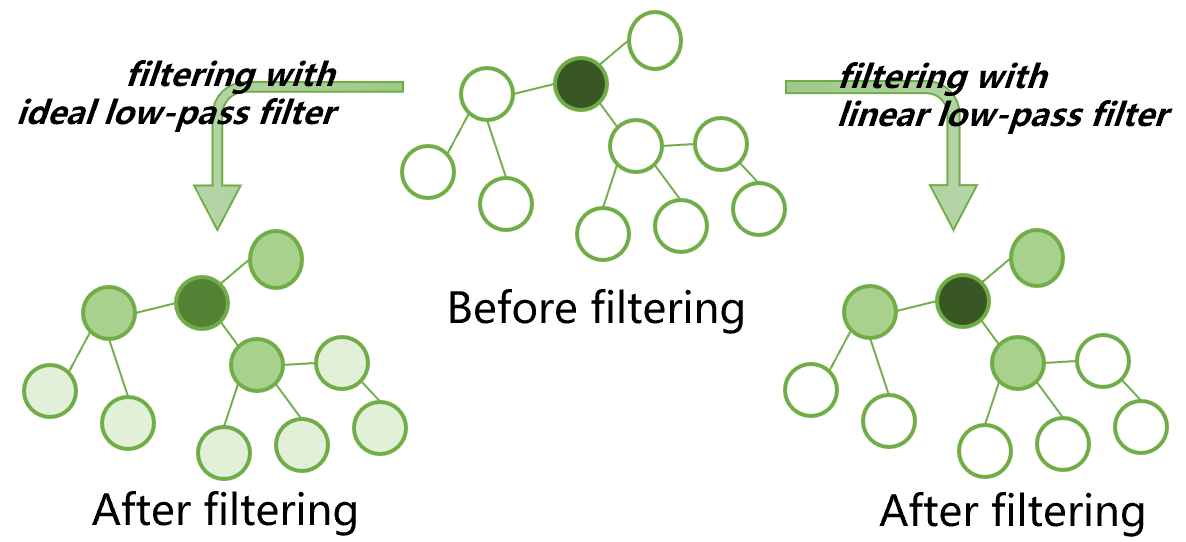}
    \caption{The ideal low-pass filter makes the signal globally smooth, and the linear low-pass filter makes the signal locally smooth but globally rough.}
    \vspace{-1em}
    \label{Fig:smooth}
\end{figure}

% 先要讲明白“为什么需要高频信息”：理想滤波器导致是“全图平滑”
An ideal low-pass filter only uses the low-frequency Fourier basis to reconstruct the spatial signal.
Since these Fourier bases are low-frequency signals in the sense of the whole graph, it means that the signal filtered by an ideal low-pass filter is also a smooth signal in the sense of the whole graph.
The left arrow in Fig. \ref{Fig:smooth} describes this global smoothing.
After filtering, the node signal spreads to the whole graph.

\begin{figure*}[!th]
		\centering
		\includegraphics[width=0.78\linewidth]{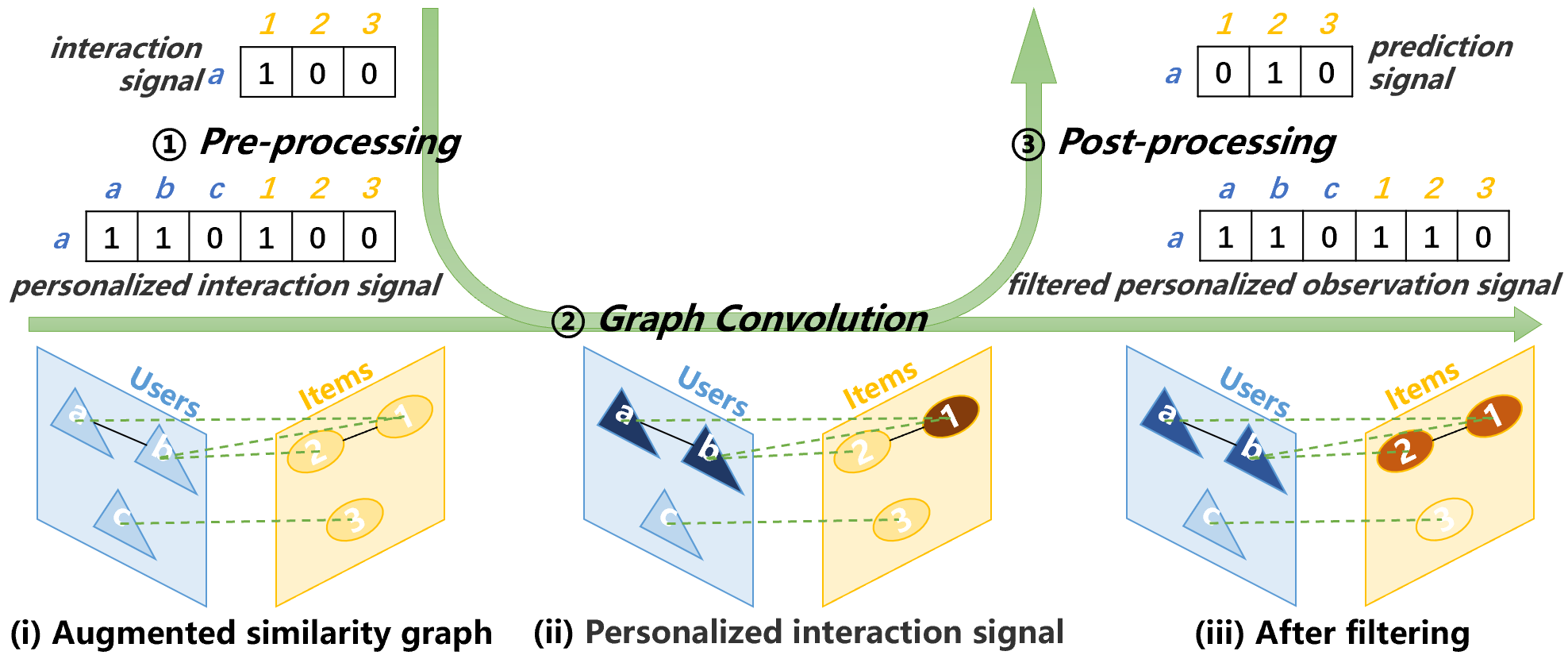}
		\caption{The proposed PGSP method. The personalized graph signal is constructed in the pre-processing step and input into the augmented similarity graph, which promotes users' potential interests.}
		\label{Fig:ours}
\end{figure*}

% 然而，由于数据的稀疏性，交互信号可能并没有足够多的信息通过这样的全图平滑来还原完整的真实偏好
However, due to the sparsity of data, the interaction signal may not have enough information to restore the complete real preference through such global smoothing.
% 相反，这样将透支了观测信号的表达能力，所以仅仅使用低频信号在一定程度上损害了用户的个性化信息
On the contrary, only using the low-frequency signals will overdraft the expressive ability of the interaction signal for the filtered signals only containing general preference information, which damages the personalized information.

% 为了解决上面提到的问题，我们实现了混合频率图滤波器
To solve the above problems, we implement a mixed-frequency graph filter, which is a low-pass filter.
% 线性低通滤波器得到的局部平滑信号
From the perspective of the spatial domain, the result of multiplication with augmented similarity matrix $A$ is neighborhood aggregation, which leads to local smoothing.
The right arrow in Fig. \ref{Fig:smooth} describes this local smoothing, after filtering, the node signal only diffuses to the directly connected nodes.
Theorem \ref{theorem.A} shows that $A$ is a linear low-pass filter with frequency response function $h(\lambda)=1-\lambda$, which indicates that the filtered signal is not globally smooth because it allows some high-frequency components of the signal to pass through.

\begin{theorem}
\label{theorem.A}
The augmented similarity matrix $A$ is a low-pass filter with a frequency response function $h(\lambda)=1-\lambda$.
\end{theorem}
\begin{proof}
First, we have:
\begin{displaymath}
  L=U\Lambda U^T=Udiag(\lambda_1,\lambda_2,...,\lambda_{m+n})U^T,
\end{displaymath}
which means that $L$ is a graph filter with a frequency response function of $h(\lambda)=\lambda$.

Then, we have:
\begin{displaymath}
  L=I-A.
\end{displaymath}

Suppose $\pmb{u}_i$ is the eigenvector of $L$ corresponds to $\lambda_i$, the derivation is as follows:
\begin{displaymath}
  L\pmb{u}_i=\lambda_i\pmb{u}_i=(I-A)\pmb{u}_i=\pmb{u}_i-A\pmb{u}_i.
\end{displaymath}
So, we have:
\begin{displaymath}
  A\pmb{u}_i=(1-\lambda_i)\pmb{u}_i,
\end{displaymath}
which means that $A$ and $L$ have the same eigenvectors, and the corresponding eigenvalues have the relationship:
\begin{displaymath}
  (\lambda_A)_i=1-\lambda_i,
\end{displaymath}
where $(\lambda_A)_i$ is the $i$-th largest eigenvalue of $A$.
Suppose:
\begin{displaymath}
\Lambda_A=diag((\lambda_A)_1,(\lambda_A)_2,...,(\lambda_A)_{m+n}).
\end{displaymath}
So, we have:
\begin{displaymath}
  A=U\Lambda_{A}U^T=Udiag(1-\lambda_1, 1-\lambda_2,...,1-\lambda_{m+n})U^T,
\end{displaymath}
which means that $A$ is a low-pass filter with a frequency response function of $h(\lambda)=1-\lambda$.
\end{proof}

% 最终的滤波器
To consider both general user preferences and personalized user preferences, we superimpose a linear low-pass filter that smooths signals locally on the basis of an ideal low-pass filter that smooths signals globally in a certain proportion as follows:
\begin{equation}
  \mathcal{H}=(1-\phi)\mathcal{H}_{ideal}+\phi A,
\end{equation}
where $\phi$ is a hyperparameter that controls the ratio of the globally smooth signal to the locally smooth signal.

% 4.3 个性化图信号处理
\subsection{PGSP Pipeline}
\label{section.PGSP}
% 两个部分
%（1）用矩阵运算代替向量运算介绍pipeline
%（2）一个实例讲清楚增广相似度矩阵和个性化交互信号
% 第一部分
In this section, we present how to combine the personalized graph signal, the augmented similarity graph and the mixed-frequency graph filter by proposing a pipeline named Personalized Graph Signal Processing which consists of %three key steps.%:
pre-processing, graph convolution and post-processing.

1) {\em Pre-processing step}. We first obtain the personalized interaction matrix by Eq. (\ref{r-tilde}).
Then we normalize it to obtain the normalized personalized interaction matrix as follows:
\begin{equation}
    \label{eq.r-norm}
    \tilde{R}_{norm}=\tilde{R} \tilde{D}^{\beta},
\end{equation}
where $\tilde{D}_{jj}=\sum_{i=1}^m\tilde{R}_{ij}$ is a diagonal matrix, and $\beta\le0$ is a hyperparameter.
Eq. (\ref{eq.r-norm}) means that the more users an item interacts with, the more likely it is a popular item, and the lower ability to express users' personalized preferences.

2) {\em Graph convolution step}. We use the proposed mixed-frequency graph filter to filter the normalized personalized graph signals to get the filtered personalized graph signals:
\begin{equation}
    \hat{\tilde{R}}_{norm}=\tilde{R}_{norm} \mathcal{H}.
\end{equation}

3) {\em  Post-processing step}. We restore the predicted personalized interaction matrix by multiplying $\tilde{D}^{-\beta}$, then the last $n$ columns are intercepted as the predicted users' preferences matrix:
\begin{equation}
  \hat{R}=(\hat{\tilde{R}}_{norm}\tilde{D}^{-\beta})_{:,-n:},
\end{equation}
where $:,-n:$ means to take all rows and last $n$ columns of the matrix.

\begin{table}\small
\centering
\caption{Statistics of the experimental data.}
\label{table.dataset}
\begin{tabular}{c|c|c|c|c}
\hline
\textbf{Dataset}     & \textbf{\#User} & \textbf{\#Item} & \textbf{\#Interaction} & \textbf{Density} \\ \hline
{\textit{Gowalla}}     & 29,858          & 40,981          & 1,027,370              & 0.084\%          \\
{\textit{Yelp2018}}    & 31,668          & 38,048          & 1,561,406              & 0.130\%          \\
{\textit{Amazon-Book}} & 52,643          & 91,599          & 2,984,108              & 0.062\%          \\ \hline
\end{tabular}
\end{table}

\begin{table*}[t]\small
\centering
\caption{Performance comparison to state-of-the-art CF methods in recent two years. \textit{RI} represents the relative improvement between PGSP and the corresponding method.}
\label{table.result}
\begin{tabular}{c|cc|cc|cc|c|cc}
\hline
\textbf{Dataset}                               & \multicolumn{2}{c|}{\textbf{Gowalla}}                               & \multicolumn{2}{c|}{\textbf{Yelp2018}}                              & \multicolumn{2}{c|}{\textbf{Amazon-book}}                           & \multicolumn{1}{c|}{\multirow{2}{*}{\textbf{\begin{tabular}[c]{@{}c@{}}Reported\\ by\end{tabular}}}} & \multicolumn{2}{c}{Avg RI}                                \\ \cline{1-7} \cline{9-10} 
Method                                         & recall                           & ndcg                             & recall                           & ndcg                             & recall                           & ndcg                             & \multicolumn{1}{c|}{}                                                                                & recall                      & ndcg                        \\ \hline
LR-GCCF~\cite{LR-GCCF}   & 0.1519                           & 0.1285                           & 0.0561                           & 0.0343                           & 0.0335                           & 0.0265                           &                                                                           \cite{SimpleX}                    & 54.88\%                     & 71.75\%                     \\
ENMF~\cite{ENMF}         & 0.1523                           & 0.1315                           & 0.0624                           & 0.0515                           & 0.0359                           & 0.0281                           &                                                                       \cite{SimpleX}              & 45.79\%                     & 47.70\%                     \\
NIA-GCN~\cite{NIA-GCN}   & 0.1726                           & 0.1358                           & 0.0599                           & 0.0491                           & 0.0369                           & 0.0287                           &                                                               \cite{NIA-GCN}~\cite{NGAT4Rec}   & 40.65\%                     & 46.80\%                     \\
LightGCN~\cite{LightGCN} & 0.1830                           & 0.1554                           & 0.0649                           & 0.0530                           & 0.0411                           & 0.0315                           &                                                                \cite{LightGCN}      & 28.95\%                     & 32.89\%                     \\
DGCF~\cite{DGCF}         & 0.1842                           & 0.1561                           & 0.0654                           & 0.0534                           & 0.0422                           & 0.0324                           &                                                                 \cite{DGCF}         & 26.94\%                     & 30.75\%                     \\
NGAT4Rec~\cite{NGAT4Rec} & -                                & -                                & 0.0675                           & 0.0554                           & 0.0457                           & 0.0358                           &                                                                \cite{NGAT4Rec}  & 30.27\%                     & 34.18\%                     \\
SGL-ED~\cite{SGL-ED}     & -                                & -                                & 0.0675                           & 0.0555                           & 0.0478                           & 0.0379                           &                                                              \cite{SGL-ED}     & 26.86\%                     & 29.57\%                     \\
DGCF~\cite{DGCF_li}     & 0.1891                           & 0.1602                           & 0.0703                           & 0.0575                           & 0.0476                           & 0.0369                           &                                                               \cite{DGCF_li}      & 17.16\%                     & 19.95\%                     \\
MF-CCL~\cite{SimpleX}    & 0.1837                           & 0.1493                           & 0.0698                           & 0.0572                           & 0.0559                           & 0.0447                           &                                                                \cite{SimpleX}       & 11.01\%                     & 13.36\%                     \\
SimpleX~\cite{SimpleX}   & 0.1872                           & 0.1557                           & 0.0701                           & 0.0575                           & 0.0583                           & 0.0468                           &                                                                 \cite{SimpleX}        & 8.47\%                      & 9.75\%                      \\
UltraGCN~\cite{UltraGCN} & 0.1862                           & 0.1580                           & 0.0683                           & 0.0561                           & 0.0681                           & 0.0556                           &                                                                 \cite{UltraGCN}    & 3.70\%                      & 3.51\%                      \\
GF-CF~\cite{GFCF}        & 0.1849                           & 0.1518                           & 0.0697                           & 0.0571                           & 0.0710                           & 0.0584                           &                                                                 \cite{GFCF}    & 1.83\%                      & 2.61\%                      \\
\hline
PGSP (ours)                                    & \textbf{0.1916} & \textbf{0.1605} & \textbf{0.0710} & \textbf{0.0583} & \textbf{0.0712} & \textbf{0.0587} &                                                                                                      & - & - \\ \hline
\end{tabular}
\end{table*}

% 第二部分
We illustrate how the PGSP method works through an example as shown in Fig. \ref{Fig:ours}.
For simplicity, we use $1$ to represent a strong node signal, and $0$ to represent a weak node signal.
Suppose user $a$ interacted with item $1$, user $b$ interacted with items $1$ and $2$, user $c$ interacted with item $3$, and user $a$ is the target user.
There is a strong similarity between the two connected nodes.
User $a$ is similar to user $b$ because they both have interacted with item $1$.
Item $1$ is similar to item $2$ because they both have interacted with user $b$.
We first construct an augmented similarity graph as shown in Fig. \ref{Fig:ours}(i), which includes not only the user-user similarity (the blue part) and item-item similarity (the yellow part), but also the user-item interaction information (the dashed line).
In pre-processing, the original interaction signal $(1,0,0)$ (indicating that user $a$ has only interacted with item $1$) concatenates the similarity signal $(1,1,0)$ (indicating that user $a$ is similar to user $b$ and himself/herself, but dissimilar to user $c$) and becomes $(1,1,0,1,0,0)$ (user $a$'s personalized graph signal).
The personalized graph signal is input into the augmented similarity graph as shown in Fig. \ref{Fig:ours}(ii).
In graph convolution, with the proposed mixed-frequency graph filter, we obtain the predicted signal, which is a combination of globally smooth signals and locally smooth signals as shown in Fig. \ref{Fig:ours}(iii).
In this process, on one hand, on the item-item similarity graph, the strong signal of item $1$ will be transmitted to item $2$, causing the signal of item $2$ to be strengthened.
On the other hand, on the user-item similarity graph, user $b$ who is similar to user $a$ has a strong signal, which will promote the signal strength of item $2$.
User $c$ who is dissimilar to user $a$ has a weak signal, which will weaken the signal strength of item $3$.
Our method makes the items that users like have a higher intensity than those that users dislike.
Finally, the predicted signal is $(0,1,0)$, which means that the model will recommend item $2$, which may be the potential interests of user $a$ because user $b$ who is similar to user $a$ has interacted with it, to user $a$.

\subsection{Comparison with GF-CF}
As a GSP-based collaborative filtering method, GF-CF~\cite{GFCF} is very similar to PGSP, and has achieved excellent results in recommendation tasks.
However, GF-CF only uses user interaction signals as input signals.
In addition, GF-CF classifies various collaborative filtering methods as filters with different frequency response functions, without analyzing the role of high-frequency components of observed signals.
In contrast, PGSP depicts user interests more accurately through the personalized graph signal containing richer user information, and constructs an augmented similarity graph containing more graph topology information, to more effectively characterize user interests.
We not only propose a mixed-frequency graph filter to introduce useful information in the high-frequency component of the observed signals by combining an ideal low-pass filter and a linear low-pass filter, but also conduct an experimental analysis on the influence of high-frequency information on recommendation accuracy.

\section{Experiments}
We conduct extensive experiments to evaluate our method.
Specifically, we evaluate the accuracy of the PGSP method, and analyze the effects of the personalized graph signal, augmented similarity graph, and mixed-frequency graph filter.

\subsection{Experimental Settings}
\begin{figure*}[t!]
\centering
\includegraphics[width=0.985\linewidth]{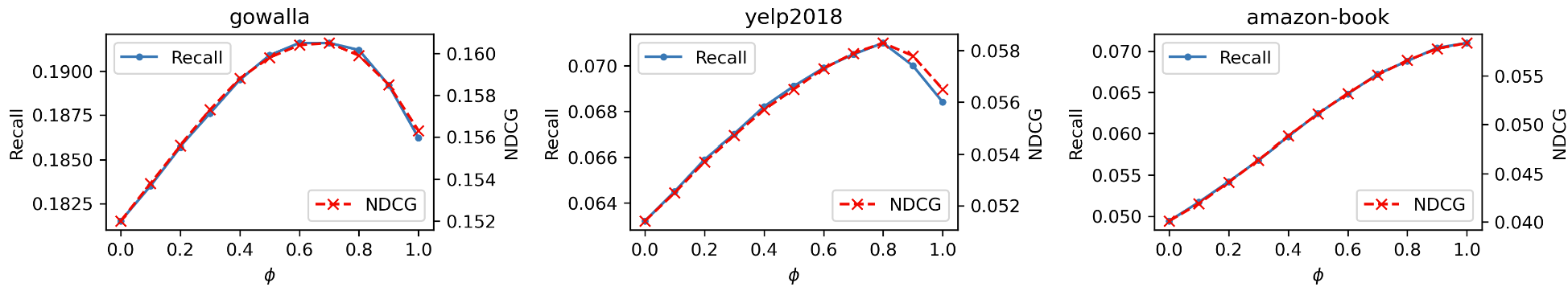}
\caption{Effectiveness of high-frequency signals on three datasets. The $x$-axis is $\phi$, which controls the weight of the high frequency (locally smooth but globally rough) information in the prediction signal. We can see that high-frequency signals can help to improve accuracy in all three datasets.}
\label{Fig.high_work}
\end{figure*}

\subsubsection{Dataset}
We use three public datasets to evaluate the recommendation accuracy of our proposed method, including  \textit{Gowalla}, \textit{Yelp2018} and \textit{Amazon-Book}, released by the LightGCN work~\cite{LightGCN}.
The statistics of the experimental data are shown in Table \ref{table.dataset}.
For each dataset, we randomly select 80\% of the historical
interactions of each user to constitute the training set, and treat the remaining as the test set.
From the training set, we randomly select 10\% of interactions as the validation set to tune hyper-parameters.

\subsubsection{Metrics}
We use two popular evaluation protocols for top-N recommendation~\cite{LightGCN}: \textit{Recall@K} and \textit{NDCG@K} with $K=20$.

\subsubsection{Baselines}
We compare the proposed PGSP method with three types of state-of-the-art methods in recent years as follows:
\begin{itemize}
\item MF-based methods, whose basic idea is to represent users and items in the same latent space with latent vectors, and use dot products as the strength of interaction possible between users and items, including ENMF \cite{ENMF}, MF-CCL \cite{SimpleX} and SimpleX \cite{SimpleX}.
\item GNN-based methods, whose basic idea is to embed users and items in the space by means of message transmission. Each user and item has an embedding vector, and the dot product is used to represent the interaction possibility of different users and items. Including LR-GCCF \cite{LR-GCCF}, NIA-GCN \cite{NIA-GCN}, LightGCN \cite{LightGCN}, DGCF \cite{DGCF}, NGAT4Rec \cite{NGAT4Rec}, SGL-ED \cite{SGL-ED},  UltraGCN \cite{UltraGCN} and DGCF \cite{DGCF_li}.
\item GSP-based GF-CF method \cite{GFCF}, which sums up different collaborative filtering methods, including those based on the neighborhood, matrix factorization and graph neural network, into low-pass filters with different frequency response functions.
\end{itemize}
We only use these latest works for comparison, and the classic baselines are omitted due to space limitations.

\subsection{Performance Comparison}
Table \ref{table.result} reports the performance comparison results and the relative improvement on each metric.
Since the datasets partition and evaluation metrics of all compared baselines are consistent with our method, we directly present the results reported in them, we can ensure that we present the optimal performance of various methods on the premise of fair comparison. The experimental setting of all baselines can be found at \url{https://openbenchmark.github.io/BARS}~\cite{zhu2022bars}.
We have the following observations from the  results:
\begin{itemize}
\item Compared with SOTA methods in recent two years, our method has a significant improvement, especially in \textit{NDCG}, which shows that our method not only optimizes the recall score, but also optimizes the ranking order.
\item As a nonparametric method, PGSP is better than recently proposed parametric methods with a large margin, especially on \textit{Amazon-Book}, a sparse dataset. This should be due to sparse data in CF, in which the data-driven methods may not work well.
\item Our method is superior to GF-CF, which is also based on GSP. This is because compared with GF-CF, our method uses the augmented similarity matrix and the personalized graph signal, and takes into account the role of high-frequency information and the normalized matrix.
\end{itemize}

\subsection{Effectiveness of High-frequency Signal}
We construct different mixed-frequency graph filters by controlling $\phi$ to change the proportion of globally smooth signals that reflect general user preferences and locally smooth but globally rough signals that reflect personalized user preferences.
The results are shown in Fig. \ref{Fig.high_work}. 
We observe that the globally smooth signal filtered only by the ideal low-pass filter ($\phi=0$) is not optimal.
The accuracy can be improved by superimposing locally smooth signals filtered by the linear low-pass filter in proportion.
Specifically, on \textit{Gowalla} and \textit{Yelp2018}, with the increase of locally smooth signals, the accuracy first increases and then decreases, and on \textit{Amazon-Book}, the accuracy keeps increasing with the more locally smooth signal.
% It shows that due to the sparse data, the observed interaction signals have insufficient expressive ability and can not fully describe users' preferences.
% Therefore, it is not enough to only use globally low-frequency signals in CF tasks.
% It shows that the expressive ability of interaction signals is too low on extremely sparse datasets.
% In this case, the locally smooth signals that reflect personalized user preferences is enough for recommendation. This experiment confirms the necessity of introducing locally smooth but globally rough signals in CF.

The density of the datasets is related to the relative performance increase obtained by adding high-frequency (locally smooth but globally rough) signals.
As mentioned earlier, the expression ability of sparse interaction signals is weak, so we introduce high-frequency signals.
According to this viewpoint, on more sparse datasets, we should observe higher performance improvement due to the introduction of high-frequency signals.
In order to give a more convincing conclusion, we carried out a new experiment and added a dataset \textit{ML-1M}~\cite{harper2015movielens}.
The density of the datasets has such a relationship: \textit{Amazon-Book} (0.062\%) < \textit{Gowalla} (0.084\%) < \textit{Yelp2018} (0.130\%) < \textit{ML-1M} (4.845\%).
Meanwhile, we conducted a 5-fold cross-validation and reported the average results to avoid the deviation caused by the dataset partition.
We observed that compared with using only low-frequency signals, the relative improvements brought by high-frequency signals are as follows: \textit{Amazon-Book} (81.25\%) > \textit{Gowalla} (17.89\%) > \textit{Yelp2018} (4.47\%) > \textit{ML-1M} (0.79\%), which empirically confirmed our viewpoint.
Since the recommender system datasets are sparse in most realistic scenarios, we believe the introduction of high-frequency signals is important.

\begin{table}\small
\centering
\caption{Training time comparison of three types of methods on the Gowalla dataset. In ENMF and LightGCN, we show the training time of one epoch, and they need multiple epochs during training. In PGSP and GF-CF, we show the total time required to complete model training.}
\label{tab.running}
\begin{tabular}{c|c|c|c|c} 
\hline
{Method}               & {PGSP} & {ENMF} & {LightGCN} & GF-CF \\ 
\hline
{Running time}         & 8m42s         & 11m5s         & 69m40s     & 5m31s        \\
{Need multiple epochs} & \xmark             & \cmark             & \cmark    & \xmark              \\
\hline
\end{tabular}
\end{table}
\subsection{Training Time Comparison}
Different from the methods based on matrix factorization (MF) and graph neural network (GNN), PGSP is a nonparametric method, which does not require parameter learning, so it can achieve very high training efficiency.
We compare PGSP with MF and GNN-based CF methods, and another GSP-based method GF-CF~\cite{GFCF} in Table \ref{tab.running}.
It should be noted that the methods based on MF and GNN require multiple epochs to obtain a trained model.
For ENMF and LightGCN, we show the training time for only one epoch, and they need to train multiple epochs until convergence.
For PGSP and GF-CF, we show the total time required to complete the model training.
It can be observed that the total time required to train PGSP is even less than the time to train one epoch in MF and GNN-based CF methods, which confirms the high efficiency of PGSP.
Due to the introduction of additional information, \ours requires a slightly longer running time than the pure GSP-based GF-CF method.

\section{Conclusion}
In this work, we further promote the performance of GSP in CF tasks.
We establish the similarity relationship from the perspective of random walk.
Based on the proposed similarity, a personalized graph signal with more personalized information is proposed to describe users more accurately, and an augmented similarity graph with more topology information is constructed to utilize the signal more effectively.
We reveal the importance of high-frequency signals in observed signals, and construct a mixed-frequency graph filter to simultaneously use the globally smooth signal and the locally smooth signal.
Finally, the effectiveness of the proposed method is analyzed through comprehensive experiments 
on three public datasets.

\begin{acks}
This work was supported by the National Natural Science Foundation of China (NSFC) under Grants 62172106 and 61932007.
\end{acks}

\bibliographystyle{ACM-Reference-Format}
\bibliography{sample-base}

\end{document}